\documentclass[12pt]{elsarticle}

\usepackage{amsthm}
\usepackage{amsmath}
\usepackage{amssymb}
\usepackage{paralist}
\usepackage[numbers]{natbib}
\usepackage{verbatim}
\usepackage{url}
\usepackage{CJK}

\def\dotsb{\ldots}

\newcommand{\phe}{\varphi}
\newcommand{\LEX}{\mathrm{LEX}}
\newcommand{\INC}{\mathrm{INC}}

\newcommand{\T}{\mathrm{T}}
\newcommand{\btt}{\mathrm{btt}}
\newcommand{\lex}{\mathrm{lex}}

\newcommand{\orr}{\quad \vee \quad}
\newcommand{\intersect}{\mathrel{\cap}}
\newcommand{\union}{\mathrel{\cup}}

\renewcommand{\complement}[1]{\overline{#1}}
\newcommand{\restr}{\mathrel{\upharpoonright\nolinebreak[4]\hspace{-0.65
ex}\upharpoonright}}

\newcommand{\lcode}{\left\langle}
\newcommand{\rcode}{\right\rangle}
\newcommand{\pair}[1]{{\lcode#1\rcode}}
\newcommand{\size}[1]{{\left|#1\right|}}

\theoremstyle{plain}        \newtheorem{thm}{Theorem}[section]
\theoremstyle{definition}   \newtheorem{defn}[thm]{Definition}
\theoremstyle{definition}   \newtheorem{nota}[thm]{Notation}
\theoremstyle{plain}        \newtheorem{prop}[thm]{Proposition}
\theoremstyle{plain}        \newtheorem{cor}[thm]{Corollary}
\theoremstyle{plain}        \newtheorem{lemma}[thm]{Lemma}
\theoremstyle{plain}        \newtheorem{ques}[thm]{Question}

\newcommand{\mc}{\mathcal}
\newcommand{\A}{\mc A}
\newcommand{\B}{\mc B}
\newcommand{\C}{\mc C}

\newcommand{\ind}{\mathrm{ind}}

\numberwithin{equation}{section}

        \usepackage{color}
 \usepackage{soul}
 \definecolor{lightblue}{rgb}{.60,.60,1}

\begin{document}
\begin{frontmatter}
\title{Things that can be made into themselves}

\author[label1]{Frank Stephan\fnref{fn1}}
\ead{fstephan@comp.nus.edu.sg}
\fntext[fn1]{Frank Stephan was supported in part by NUS grant R252-000-420-112.}
\address[label1]{National University of Singapore}

\author[label2]{Jason Teutsch}
\ead{teutsch@cse.psu.edu}
\address[label2]{Penn State University}

\begin{abstract}
\noindent
One says that a property $P$ of sets of natural numbers can be made into
itself iff there is a numbering $\alpha_0,\alpha_1,\ldots$ of all left-r.e.\
sets such that the index set $\{e: \alpha_e$ satisfies $P\}$ has the property
$P$ as well. For example, the property of being Martin-L\"of random can
be made into itself.  Herein we characterize those singleton
properties which can be made into themselves.
A second direction
of the present work is the investigation of the structure of left-r.e.\ sets
under inclusion modulo a finite set.  
In contrast to the corresponding structure for r.e.\ sets, which has
only maximal but no minimal members,  both minimal and maximal
left-r.e.\ sets exist.  Moreover, our construction
of minimal and maximal left-r.e.\ sets greatly differs from Friedberg's
classical construction of maximal r.e.\ sets. 
Finally, we investigate whether the properties of minimal and maximal
left-r.e.\ sets can be made into themselves.
\end{abstract}

\end{frontmatter}


\section{Introduction}

\noindent
The roots of recursion theory entwine self-reference.
Even before Turing \cite{Tur36} formalised the notion of computation
using machines with infinite tapes and finite control structures,
mathematicians conceived of primitive recursive and recursive
functions \cite{Ack28,God31}. In particular, Kurt G\"odel captured
many aspects of computation in his formal system of arithmetic and
exploited its self-referential properties in the proof of his famous
incompleteness theorem \cite{God31}.  In order to show that the theory
of natural numbers does not have a consistent and
complete r.e.\ axiomatization, G\"odel created a first-order formula
which informally states, with respect to an underlying primitive-recursive
set of axioms,
\begin{quote}
\emph{``This statement is unprovable.''}
\end{quote}
so that neither the statement nor its negation has a mathematical
proof with respect to the given set of axioms.
G\"odel's ground-breaking construction contained various important
concepts including coding, or numbering, techniques.
For this reason the acceptable numberings are
also called, after him, \emph{G\"odel numberings}.
The expressive strength of a general-purpose computer language
is precisely what makes G\"odel's self-referential statement possible. 
Self-reference has manifested itself in computer science and
mathematics in the form of fixed point theorems, in particular
Kleene's Recursion Theorem \cite{Kle38}, Roger's Fixed-Point
Theorem~\cite[Theorem~11-I]{Rog67}, the Arslanov Fixed Point
Theorem~\cite{ANS77} and its generalizations~\cite{Ars81, Ars89,
JLSS89}, as well as other diagonalization methods~\cite{Odi89, Soa87}.
Today research continues
in the area of machine self-reference and self-knowledge~\cite{CM09}.

In recursion
theory one often studies effective listings of r.e.\ sets and partial-recursive
functions. On one hand there are the acceptable numberings introduced by
G\"odel \cite{God31}; on the other hand Friedberg \cite{Fri58} showed that
there are also one-one numberings of the above named objects where each item
occurs exactly once. In this paper, we look at self-reference in
terms of numberings of left-r.e.\ sets. Here a set $A$ is left-r.e.\ iff it
can be approximated by a uniformly recursive sequence of sets such that
$A$ is the lexicographic supremum of all these sets. Furthermore, a
left-r.e.\ numbering
$A_0,A_1,\ldots$ is an effective sequence of left-r.e.\ sets as defined
more precisely below (Definition~\ref{def:left-r.e.-numbering}).
Numberings for left-r.e.\ sets were first studied by
Brodhead and Kjos-Hanssen \cite{BK09,KST10} and
provide more expressive possibilities than the
traditional numberings for r.e.\ sets.

In this paper, we are
especially interested in classes $\cal C$ of sets such that
there is a numbering of all left-r.e.\ sets in which
the index set of the left-r.e.\ members of $\cal C$ is itself a
member of the class $\cal C$.  For some reason such things exist, and
we call this
phenomenon ``things that can be made into themselves.'' Here the
phrase ``can be made into themselves'' implicitly refers to the fact
that we permit ourselves the flexibility to choose the underlying
numbering of all left-r.e.\ sets for a desired purpose. If we would
follow the usual default and use only acceptable left-r.e.\ numberings,
this would prohibit many things from being made into themselves. Indeed Rice's
Theorem \cite{Ric53} holds for acceptable left-r.e.\ numberings and
therefore any non-trivial index set is many-one hard either for the
halting problem or its complement.

In 1958, Friedberg \cite{Fri58} constructed a maximal r.e.\ set, that
is, an r.e.\ set maximal under inclusion up to finite differences,
thereby bringing Post's program \cite{Odi89,Pos44,Soa87} to an abrupt halt.
Post \cite{Pos44} had wished to prove the existence of Turing incomplete
r.e.\ sets by building sets with sparse complements.  Maximal sets can be
Turing complete and have the thinnest possible complements for r.e.\
sets, so Friedberg's result shows that ``thinness'' alone cannot
achieve Turing incompleteness.  As independently discovered by
Friedberg \cite{Fri57} and Muchnik \cite{Soa87}, Turing incomplete r.e.\ sets do
exist by alternate methods.   In Section~\ref{sec:minmax} we
introduce the concept of maximal and minimal left-r.e.\ sets.  Unlike
the class of r.e.\ sets, which has only maximal sets, both minimal and
maximal left-r.e.\ sets exist (Theorem~\ref{thm:zulu}).  Maximal r.e.\
sets cannot be maximal left-r.e (Theorem~\ref{thm: maxsep}), and among
the minimal and maximal left-r.e.\ sets only singleton maximal
left-r.e.\ sets can be made into themselves (Theorem~\ref{thm:lowerfarm}). 

We shall show that the Martin-L\"of random sets and 1-generic sets
can be made into themselves (Corollary~\ref{randoms are random2} and
Corollary~\ref{cor: bobo}), though not at the same time
(Proposition~\ref{prop: ML 1-generic calamity}), whereas the r.e.,
co-r.e.\ and recursive sets each cannot be (Corollary~\ref{cor:
pigpong}).  We characterise the left-r.e.\ sets whose index sets can
be made equal to the set itself (Theorem~\ref{thm: bambam}) and
discuss the complexity of the inclusion problem for left-r.e.\
numberings (Theorem~\ref{thm:left-r.e. inclusion}).

\begin{nota} A \emph{numbering} $\phe$ of partial-recursive functions is a
mapping $e \mapsto \phe_e$ such that the induced mapping $\pair{e,x}
\mapsto \phe_e(x)$ is partial-recursive. $W_e^\phe$ denotes the domain of
$\phe_e$ and we may omit the superscript when it is clear from context.
We identify numbers in a one-one way with binary strings so that the
ordering of the numbers is translated into the length-lexicographic
ordering of the strings.  We use $\size{e}$ to denote the
length of the string $e$, and we shall appeal to the fact that
$\size{e} \leq 1+ \log e$ for all~$e>0$.

Let a machine $\psi$ be a partial-recursive mapping from strings
to strings. The complexity of $x$ with respect to $\psi$,
called $C_\psi(x)$ is the length of the
shortest input $y$ with $\psi(y) = x$. $\psi$ is called universal
iff its range contains all strings and for every further 
machine $\phe$ there is a constant $c$ such that
for all $y$ in the domain of $\phe$, $C_\psi[\phe(y)] \leq |y|+c$.
It the field of Kolmogorov complexity, one fixes some plain universal
machine and denotes with $C(x)$ the plain Kolmogorov complexity of $x$
with respect to this machine. Similarly, one can consider prefix-free
machines where a machine $\psi$ is prefix-free iff any two strings
in its domain, neither of the two is a proper prefix of the other one.
One can then define the prefix-free Kolmogorov complexity $H$
as above with respect to a fixed machine which is universal among
all prefix-free machines. Calude \cite{Cal94} and Li and Vit\'anyi \cite{LV08}
provide further background on Kolmogrov complexity.

Let $A \triangle B$ denote the symmetric difference
of $A$ and $B$, that is, $A \cup B - A \cap B$. Furthermore,
$\overline{A} = {\mathbb N}-A$ is the complement of the set $A$.
Furthermore, $A \subseteq^* B$ means that almost all elements of $A$
are also in $B$ and $A \subset^* B$ means that in addition to the previous,
there are infinitely many elements in $B-A$.
For finite strings $\sigma$ and $\tau$, $\sigma \cdot
\tau$ denotes concatenation of $\sigma$ and $\tau$, $\sigma
\sqsupseteq \tau$ means $\sigma$ extends $\tau$ and $\sigma
\sqsubseteq \tau$ means $\sigma$ is a prefix of $\tau$.  Similarly for
sets, $\sigma \sqsubseteq A$ means that $\sigma$ is a prefix of $A$
(where, as usual, the set $A$ is identified with the infinite sequence
$A(0)A(1)\ldots$ given by its characteristic function).
A set is recursively enumerable (or just \emph{r.e.}) iff it is either
empty or the range of a recursive function.
A set is \emph{co-r.e.}\ if it is the complement of an r.e.\
set, $'$ is the jump operator and $\equiv_\T$
is Turing equivalence.  We say $A$ is \emph{$B$-recursive} if $A
\leq_\T B$. $A \leq_\btt B$ if membership in $A$ can be decided by
uniformly constructing a Boolean formula over finitely many variables
and evaluating it using membership values from $B$.  For a set $A$, we
use $A \restr n$ to denote the prefix of $A$'s characteristic sequence
$A(0)A(1) \dotsb A(n)$.  A subset of
natural numbers is $\Pi^0_n$ if it can be described by a formula
consisting of $n$ alternating quantifiers, starting with a universal
quantifier, and ending with a recursive predicate. Furthermore,
a set is $\Sigma^0_n$ iff its complement is $\Pi^0_n$.

A set $A$ is called \emph{autoreducible} \cite{Tra70}
if for all $x$, whether $x$ is a member of $A$ can be effectively
determined by querying $A$ at positions other than $x$; a set $A$ is
called \emph{strongly infinitely-often autoreducible} \cite{Ars00} if there is
a partial-recursive function $\psi$ such that for all inputs of
the form $x = A(0)\ldots A(n-1)$, either $\psi(x)$ outputs $?$ or
$\psi(x)$ outputs $A(n)$ and the latter happens infinitely often;
note that there are strongly infinitely often autoreducible sets which
are not autoreducible.
For any numbering $\alpha$, the \emph{$\alpha$-index set} of a class $\C$
is the set $\{e: \alpha_e \in \C\}$.  For sets of nonnegative integers
$A$ and $B$, \emph{$A \leq_\lex B$} means that either $A = B$ or the
least element $x$ of the symmetric difference satisfies $x \in B$.  A
set $A$ is \emph{left-r.e.}\ iff there is a uniformly recursive
approximation $A_0,A_1,\ldots$ of $A$ such that $A_s \leq_{\lex}
A_{s+1}$ for all $s$. The symbol
$\oplus$ denotes join.  For further background on recursion theory and
left-r.e.\ sets, see the textbooks of Calude \citep{Cal94},
Downey and Hirschfeldt \citep{DH10}, Li and Vit\'anyi \citep{LV08},
Nies \citep{Nie09}, Odifreddi \citep{Odi89,Odi99}, Rogers \citep{Rog67}
and Soare~\citep{Soa87}.
\end{nota}

\noindent
The reader may already be familiar with left-r.e.\ reals, which admit
an increasing, recursive sequence of rationals from below, however in
the context of effective enumerations it makes more sense to consider
left-r.e.\ \emph{sets}, see \cite[Section~2]{KST10}.  For example, the
infinite left-r.e.\ sets have an left-r.e.\ numberings while the
coinfinite left-r.e.\ sets do not have one; if one would only consider
left-r.e.\ reals, the distinction between coinfinite and infinite sets
would disappear and so, for example, the coinfinite reals would have
a left-r.e.\ numbering. So the results depend a bit on the setting (sets
versus reals) and we decided to follow the more natural setting of sets
(as most of recursion theory does).

\begin{defn} \label{def:left-r.e.-numbering}
A \emph{left-r.e.\ numbering} $\alpha$ is a mapping
from natural numbers to left-r.e.\ sets given as the limits
of a uniformly recursive sequences in the sense
\[
e \mapsto \lim_{s \to \infty} \alpha_{e,s} = \alpha_e
\]
where the following two conditions hold:
\begin{enumerate}[\scshape (i)]
\item the mapping $e,s,n \mapsto \alpha_{e,s}(n)$ is recursive
  and $\{0,1\}$-valued;
\item $\alpha_{e,s} \leq_\lex \alpha_{e,s+1}$ for all $s$.
\end{enumerate}
A left-r.e.\ numbering is called \emph{universal} if
its range includes all left-r.e.\ sets, and a left-r.e.\ numbering
$\alpha$ is called an ($K$-)\emph{acceptable} left-r.e.\ numbering if
for every left-r.e.\ numbering $\beta$ there exists a ($K$-)recursive
function $f$ such that $\alpha_{f(e)} = \beta_e$ for all $e$.  Here
$K$ denotes the halting set.
\end{defn}

\noindent
Acceptable numberings permit an
effective means for coding any algorithm, so an example of an
acceptable numbering can be obtained by the functions defined in some
general purpose programming language where some adjustments in definitions
have to be made, for example, that variables take as values natural numbers
and that there is exactly one input and one output and that there are no
constraints on the size of the numbers stored in the variables; furthermore,
the program texts have to be identified with natural numbers coding them
and ill-formed programs just correspond to the everywhere undefined function.

\begin{defn}
We say that a class of sets $\C$ \emph{can be made into itself} if
there exists a
universal left-r.e.\ numbering $\beta$ such that
\[
\{e : \beta_e \in {\mc C}\} \in {\mc C}.
\]
\end{defn}

\noindent
Note that in this context there are classes $\mc C$ which can be made into
themselves and which do not entirely consist of left-r.e.\ sets.
This will be essential for various results; for example the
Martin-L\"of random sets can be made into themselves
(Corollary~\ref{randoms are random2}) while
the Martin-L\"of random left-r.e.\ sets cannot be made into
themselves (Proposition~\ref{prop:left-r.e. ML not=self}). Hence
permitting $\mc C$ to have members which are
not left-r.e.\ is often necessary and is also natural in the
case for many classes.

Our primary tool for making things into themselves will be indifferent
sets.  An indifferent set is a list of indices where membership in a
given set can change without affecting membership in some class.

\begin{defn}[Figueira, Miller and Nies \cite{FMN09}]
An infinite set $I$ is called \emph{indifferent for a set~$A$ with
respect to $\C$} if for any set $X$,
\[
X \triangle A \subseteq I \implies X \in \C.
\]
When the class $\C$ is clear from context, we may omit it.
\end{defn}

\section{Classes that can be made into themselves}

\noindent
We show that any class of nonrecursive sets which
either contains the Martin-L\"of random sets or
contains the weakly 1-generic sets can be made into itself.
Our proof relies crucially on co-r.e.\ indifferent sets which
are retraceable by recursive functions.

A set $A$ is called \emph{Martin-L\"of random} \citep{ML66,Sch71} if there
exists a constant $c$ such that for all $n$, $H(A \restr n) \geq n -
c$.  Intuitively, $A$ is random if every prefix of $A$ is
incompressible and therefore lacks a simple pattern.
Zvonkin and Levin~\cite{ZL70} and later Chaitin~\citep{Cha87} gave an
example of a left-r.e.\ Martin-L\"of random real
called~\emph{$\Omega$}.

Figueira, Miller and Nies \cite{FMN09} constructed indifferent sets
for the class of Martin-L\"of random sets.  One of their approaches
is to build indifferent sets for non-autoreducible sets.  While this
works for Martin-L\"of random sets, the technique does not
generalise to weaker forms of randomness because recursively random
sets may be autoreducible \cite{MM04}. On the other hand,
Franklin and Stephan \cite{FS12}
showed that every complement of a dense simple set is indifferent
with respect to Schnorr randomness for all Schnorr random sets.
The arguments in Lemma~\ref{K-recursive implies retraceable} and
Theorem~\ref{indifferent and retraceable for random} are also
essentially due to Figueira, Miller and Nies \cite{FMN09}, however we
find it useful to make explicit the property of retraceability.

\begin{defn}
A set $A = \{a_0,a_1,a_2 \ldots \}$ is \emph{retraceable} if there
exists a partial-recursive function $f$ satisfying $f(a_{n+1}) = a_n$
for all $n$ and $f(x) < x$ whenever $f(x)$ is defined.
A set $S$ is \emph{approximable} if there exists an $n$ and a
recursive function $f$ such that for any $x_1,\ldots,x_n$
with $x_1 < \dotsb < x_n$, $f(x_1,\ldots,x_n) \in \{0,1\}^n$ and
$f(x_1,\ldots,x_n)$ agrees with the characteristic vector
$(S(x_1), \ldots, S(x_n))$ in at least one place. More generally,
if agreement in not only one but $m$ places is required, we
say $S$ is \emph{$(m,n)$-recursive}, where $1 \leq m \leq n$.
\end{defn}

\begin{lemma} \label{K-recursive implies retraceable}
For every $K$-recursive function $f$, there exists a co-r.e.\ set $I =
\{i_0, i_1, i_2, \dotsc \}$ which is retraceable by a recursive
function and satisfies $f(n) < i_n < i_{n+1}$ for all $n$.
\end{lemma}

\begin{proof}
Let $\{f_s\}$ be a recursive approximation to $f$ satisfying $\max f_s
< s$.  We construct $I$ by a movable marker argument.  The set
\[
I_s =\{i_{0,s}, i_{1,s}, i_{2,s}, \dotsc \}
\]
will be a recursive approximation to $I$ at stage $s$.  Set $I_0 =
\omega$.  At stage $s+1$, choose the least $n$ satisfying $f_{s}(n)
\neq f_{s+1}(n)$ and enumerate sufficiently many elements into
$\complement{I}_{s+1}$ such that
 \begin{itemize}
\item For all $k \geq n$, $i_{k, s+1} \geq s+1$, and
\item For all $k < n$, $i_{k, s+1} = i_{k,s}$.
\end{itemize}
For each $n$, $\{f_t(n)\}$ settles in some stage $s_n+1$ and so
\begin{equation*}
i_n = i_{n, s_n} \geq s_n + 1 > f(n).
\end{equation*}
Furthermore, the recursive function
\[
g(x)=
\begin{cases}
i_0 &\text{if $x \leq i_1$, and} \\
\max I_{x+1} \intersect\{0, 1, 2, \dotsc x-1\} &\text{otherwise.}
\end{cases}
\]
witnesses that $I$ is retraceable because if $I$ differs from
$I_{x+1}$ at some index below $x+1$, then by construction $x \notin
I$.\end{proof}
\noindent
The set in Lemma~\ref{K-recursive implies retraceable} is retraced by
a total recursive function.
Hence there is a recursive function $h$ which maps $I$ surjectively
to the set of natural numbers. In the above case, one can also see
directly that such a $h$ exists, as one can choose $h$ as
$$
  h(x) = |I_{x+1} \cap \{0,1,\ldots,x\}|
$$
and then $h$ has the desired property $h(i_n) = n$. A set which is
retraceable by a recursive function is $(1,2)$-recursive \citep{ST10},
and therefore the set $I$ above is also approximable.

\begin{lemma} \label{K-recursive into itself}
Let $\cal C$ be a class of nonrecursive sets containing:
\begin{enumerate}[\scshape (i)]
\item  a $K$-recursive member $A$ with a co-r.e.\ and retraceable set
$I$ which is indifferent for $A$ with respect to $\cal C$ and
\item a left-r.e.\ set $X = \sup X_s$ such that all the recursive
approximations $X_s$ to $X$ satisfy $\sigma \cdot X_s \notin \C$ while
$\sigma \cdot X \in \C$ for all strings~$\sigma$.
\end{enumerate}
Let $\cal D$ be a superclass of $\cal C$ not containing any recursive
set.  Then there exists a $K$-acceptable universal left-r.e.\
numbering which makes $\cal D$ into itself.
\end{lemma}

\begin{proof}
Let $i_0,i_1,i_2,\ldots$ be the elements of $I$ in ascending
order and let the numbering $\alpha_0,\alpha_1,\alpha_2,\ldots$ be an acceptable
universal left-r.e.\ numbering. Recall that there is a recursive
function $h$ with $h(i_n) = n$ for all $n$. Let $A_s$ be an approximation
of $A$ in the limit. Now define
$$
  \beta_e =
\begin{cases}
     \alpha_{h(e)} & \text{if $e \in I$,} \cr
     \sigma_e \cdot X_s      & \text{if $e \notin I$ and $s$ is the
largest stage
                     with $A_s(e)=0$ and} \cr
     \sigma_e \cdot X       & \text{if $e \notin I$ and $e \in A$.}
\end{cases}
$$
where $\sigma_e$ is a string chosen when $e$ is enumerated into the
complement of $I$ at some stage $s$ such that $\sigma_e >_\lex
\alpha_{h(e),s}$.  Each $\beta_e$ is left-r.e.\ because $h$ is recursive, the
complement of $I$ is r.e.\ and $\gamma = \sup_s \gamma_s$. Furthermore,
$\beta$ is a $K$-acceptable numbering as the mapping $n \mapsto i_n$
is $K$-recursive. For $e \notin I$, $\beta_e \in {\cal D}$ iff
$\beta_e \in {\cal C}$
iff $e \in A$.  As $I$ is indifferent for~$A$ with
respect to $\cal C$,
it follows that $\{e: \beta_e \in {\cal D}\}$ is in $\cal C$ and therefore
also in $\cal D$. So $\cal D$ is made into itself by the universal
left-r.e.\ numbering~$\beta$.
\end{proof}

\noindent
A set $A$ is called \emph{low} if $A' \equiv_\T K$ and $A$ is called
\emph{high} if $A' \geq_\T K'$.

\begin{thm} \label{indifferent and retraceable for random}
For every low Martin-L\"of random set $A$, there exists a co-r.e.\
set which is indifferent for $A$ with respect to the class of
Martin-L\"{o}f random sets
and retraceable by a recursive function.
\end{thm}

\begin{proof}
Let $A$ be a low Martin-L\"of random set, for example
\begin{equation} \label{eqn:2xOmega}
A = \{x : 2x \in \Omega\}
\end{equation}
is Martin-L\"of random and low by van Lambalgen's Theorem
\citep{vL87} and \cite[Theorem~3.4]{DHMN05}, see also
\cite[Theorem~3.4.11]{Nie09}.  Then
\[
f(n) = \max \{m : H(A \restr m) \leq m + 3n\}
\]
is partial-recursive in $A$ and hence $K$-recursive.  By
Lemma~\ref{K-recursive implies retraceable}, there exists a co-r.e.
set $I$ which is retraceable by a recursive function and satisfies
\begin{equation} \label{eqn: random in dominates f(n)}
f(n) < i_n < i_{n+1}
\end{equation}
for all n.  Let $k(m)$ be the number such that
\[
i_{k(m)} < m \leq i_{k(m)+1},
\]
and let $r(m)$ be the number such that
\[
f[r(m)] < m \leq f[r(m)+1],
\]
which exists by Miller and Yu's
Ample Excess Lemma \cite{MY08}, see \cite[Corollary~6.6.2]{DH10}.  By
\eqref{eqn: random in dominates f(n)} we have $k(m) \leq r(m)$ for all
sufficiently large $m$; otherwise
\[
f[r(m)+1] < i_{r(m)+1} \leq i_{k(m)} < m,
 \]
which is impossible.   

Suppose that there were some Martin-L\"of non-random set $N$ such that $N
\triangle A \subseteq I$. We can code a prefix of the set $A$ given
sufficiently long prefixes for $N$ and $I$, and so for infinitely many $m$
\begin{align*}
H(A \restr m) &\leq H(N \restr m) + H[A(i_0)A(i_1)\dotsc A(i_{k(m)})]
+ 2\log m + O(1)  \\
 &< m + 2k(m) + 2\log  m + O(1) \\
 & \leq m + 2r(m) + 2\log  m + O(1).
\end{align*}
Here the additive $\log$ factor is used for
coding two implicit programs into a single string. On the other hand,
by the definition of $f$,
\[
H(A \restr m) >  m + 3r(m)
\]
for all $m$, a contradiction.  Therefore $I$ is indifferent for $A$.
\end{proof}

\noindent
We are now ready to prove that several classes can be made into themselves.
Since left-r.e.\ Martin-L\"of random sets exist \citep{Cha87,DH10},
the following result is immediate from Theorem~\ref{indifferent and
retraceable for random} and Lemma~\ref{K-recursive into itself}.

\begin{cor} \label{randoms are random2}
If a class $\C$ contains all Martin-L\"of random sets
and no recursive sets then $\C$ can be made into itself.
In particular, the classes of Martin-L\"of random sets, recursively
random sets, Schnorr random sets, Kurtz random sets, bi-immune sets, immune
sets, sets which are not strongly infinitely often autoreducible
and nonrecursive sets can be made into themselves.
\end{cor}

\noindent
See the usual textbooks on recursion theory and algorithmic randomness
for the definition of these notions \cite{DH10,LV08,Nie09,Odi89,Rog67,Soa87}
and the paper of Arslanov \cite{Ars00} for the definition of the strongly
infinite-often autoreducible sets.
It is also straightforward to make non-random sets into themselves via
an acceptable numbering:
Just enumerate the left-r.e.\ sets on the even indices and one fixed member
of the class on the odd indices. This numbering makes each class
containing all non-immune sets (plus perhaps some others) have a
non-immune index set.

\bigskip
\noindent
We now investigate self reference for the class of 1-generic sets, a
class of sets orthogonal to Martin-L\"of random sets with respect to
Baire category and measure.  A set of binary strings $A$ is called
\emph{dense} if for every string $\sigma$ there exists $\tau \in A$
extending $\sigma$.  A set is \emph{weakly 1-generic} if it has a
prefix in every dense r.e.\ sets of binary strings.  Furthermore $X$
is \emph{1-generic} if for every (not necessarily dense) r.e.\ set of
strings $W$, either $X$ has a prefix in $W$ or some prefix of $X$
has no extension in $W$.  Every 1-generic set is weakly
1-generic \citep{Nie09}.  The following result isolates and generalises
the main idea of \citep[Theorem~23]{JST11}.

\begin{thm} \label{thm:1-generic co-r.e. indifferent}
Every $K$-recursive 1-generic set $A$ has a co-r.e.\ indifferent set
which is retraceable by a recursive function.
\end{thm}

\begin{proof}
Let $W_0, W_1, \dotsc$ be any enumeration of the r.e.\ sets, and let
$R_e$ denote the $e^\text{th}$ genericity requirement: \emph{$\rho$
satisfies $R_e$} if either some prefix of $\rho$ belongs to $W_e$ or
no proper extension of $\rho$ belongs to $W_e$.  First we show that
there exists a $K$-recursive function $f$ such that
\begin{multline*}
(\forall n)\: (\forall e \leq f(n))\: (\forall \sigma\in\{0,1\}^{f(n)})\: \\
\left[\text{$\sigma \cdot A[f(n)]A[f(n)+1] \dotsb A[f(n+1)]$ satisfies
$R_e$}\right].
\end{multline*}
For any given $\sigma$ and $e$, there must be some sufficiently long
segment of $A$, say $A(\size{\sigma})A(\size{\sigma}+1)\dotsb A(c_{\sigma,e})$,
satisfying $R_e$ since $W_e$ is an r.e.\ set and $A$ is 1-generic. 
Now let $f(0) = 0$ and
\[
f(n+1) = \max\{ c_{\sigma,e} : \size{\sigma},e \leq f(n)\}.
\]
$f$ can be computed using an $A$ and a halting set oracle, hence $f$
is $K$-recursive.  Now using Lemma~\ref{K-recursive implies
retraceable}, obtain a co-r.e. set $I$ which is retraceable by a
recursive function and satisfies $i_n > f(2n)$ for all $n$.  By the
pigeonhole principle, for every $n$ there exist at least $n$ intervals
below $f(2n)$ of the form
\[
J_k = \{f(k)+1, f(k)+2, \dotsc, f(k+1)\} \quad (k \leq 2n)
\]
which do not contain a member of $I$.  Hence $J_n \intersect I =
\emptyset$ for infinitely many $n$.  For any $B$ satisfying $A
\triangle B \subseteq I$, each
such $n$ witnesses that some
initial segment of $B$ satisfies $R_e$ for all $e \leq
f(n)$, hence $I$ is indifferent for $A$ with respect to the class of
1-generic sequences.
\end{proof}

\noindent
While a left-r.e.\ set cannot be 1-generic \citep{Nie09}, it can be
weakly 1-generic \citep{Soa87}.  This follows from the fact that a
1-generic set cannot compute a nonrecursive r.e.\ set \citep{Soa87}. 
Thus by Theorem~\ref{thm:1-generic co-r.e. indifferent} and
Lemma~\ref{K-recursive into itself}, we obtain the following result.

\begin{cor} \label{cor: bobo}
Any class of non-recursive sets containing the weakly 1-generic sets
can be made into itself.
\end{cor}

\noindent
Day has thoroughly investigated indifferent sets for 1-generic
sets~\cite{Day11}. He showed that every 1-generic set has an
indifferent set which is itself 1-generic and also points out, as
follows from Theorem~\ref{thm:1-generic co-r.e. indifferent}, that
every $K$-recursive 1-generic set has a co-r.e.\ indifferent set.

\section{Things which cannot be made into themselves}

\noindent
In this section we show that there are many classes which cannot
be made into themselves. The easiest example is the class of all finite
sets as this class cannot have a finite index set.

\begin{thm} \label{thm:can't enumerate the non-r.e. sets}
There is no left-r.e.\ numbering for the non-r.e.\ left-r.e.\ sets. 
Similarly, there is no left-r.e.\ numbering for the non-recursive
left-r.e.\ sets.
\end{thm}

\begin{proof}
Assume $\alpha_0,\alpha_1,\ldots$ is a recursive enumeration containing
no cofinite set. It is now shown that there is also a non-r.e.\ left-r.e.\
set $B$ which differs from all $\alpha_e$. For this, let $F$ be
the $K$-recursive function such that $F(e)$ is the maximum of the
$e$-th non-elements in each of the sets $\alpha_0,\alpha_1,\ldots,\alpha_e$.
One builds $B$ such that the complement of $B$ consists of
elements $x_e = 2^e \cdot 3^{d(e)}$ where $d(e)$ is the supremum of all
$F_s(e)$ for a recursive approximation $F_s$ to $F$; furthermore,
whenever $x_e \notin W_{e,s} \wedge 3x_e \in W_{e,s}$
then $d(e)$ is incremented by $1$. Note that the latter is done only
once after $F(e)$ has converged and that the latter enforces that
$W_e(x_e) \neq B(x_e) \vee W_e(3x_e) \neq B(3x_e)$
so that $B$ is not an r.e.\ set.
It is easy to see that $B$ is a left-r.e.\ set; the reason is that
the definition of $d(e)$ permits to make an approximation $x_{e,s}$ to $x_e$ 
monotonically from below and that therefore the approximation
$B_s = \{y: \forall e\,[y \neq x_{e,s}]\}$ is a left-r.e.\ approximation
to $B$. Hence $\alpha_0,\alpha_1,\ldots$ can neither be the numbering
of all nonrecursive left-r.e.\ sets nor the numbering of all
non-r.e.\ left-r.e.\ sets.
\end{proof}

\noindent
Although somewhat disappointing, the next fact follows as a consequence.

\begin{cor} \label{cor: pigpong}
The r.e.\ sets, co-r.e.\ sets and recursive sets cannot be made into
themselves.
\end{cor}

\begin{proof}
Suppose that $\alpha$ is a universal left-r.e.\ numbering which makes
the r.e.\ sets into themselves, and say the $\alpha$-index set of the
r.e.\ sets is $R$.  Let $X$ be any set which is left-r.e.\ but not
r.e., for example a left-r.e.\ Martin-L\"of random.  Now define a
left-r.e.\ numbering $\beta$ by
\[
\beta_e =
\begin{cases}
\alpha_e &\text{if $e \notin R$,} \\
\sigma\cdot X \text{ for some finite $\sigma$} &\text{otherwise.}
\end{cases}
\]
In detail, $\beta_e$ follows the enumeration of $\alpha_e$ until $e$
gets enumerated into $R$ (if this ever happens), at which point
$\beta$ switches to enumerating $X$.  Thus $\beta$ is an enumeration
of the non-r.e.\ left-r.e.\ sets, contrary to Theorem~\ref{thm:can't
enumerate the non-r.e. sets}.

Now, suppose that some universal left-r.e.\ numbering $\gamma$
makes the co-r.e.\ sets into themselves.  Let $Q$ be the
$\gamma$-index set of the co-r.e.\ sets, and  note that the class of
left-r.e.\ co-r.e.\ sets is the class of left-r.e.\ recursive sets.
By a construction analogous to the one for $\beta$ above, there exists a
left-r.e.\ numbering consisting of the left-r.e.\ sets with $\gamma$-indices
in $\complement{Q}$. This is an enumeration of all left-r.e.\ sets which
are non-recursive, contradicting Theorem~\ref{thm:can't enumerate the
non-r.e. sets}. Since $Q$ is also the index set of recursive sets,
the recursive sets cannot be made into themselves either.
\end{proof}

\noindent
Another example of what cannot be done is the following; the class
of left-r.e.\ Martin-L\"of random sets is quite natural and also
known as the class of $\Omega$-numbers \cite{CHKW01,KS01,MT13}.
The reason is that they can be represented as the halting probability
of some universal prefix-free Turing machine.

\begin{prop} \label{prop:left-r.e. ML not=self}
The left-r.e.\ Martin-L\"of random sets cannot be made into themselves.
\end{prop}

\begin{proof}
If the left-r.e.\ Martin-L\"of random reals could be made into
themselves, then the set of indices for Martin-L\"of non-random
reals would be $\Delta_2$ inside this numbering.  This contradicts a
theorem of Kjos-Hanssen, Stephan, and Teutsch \citep{KST10} which says
that the Martin-L\"of non-randoms are never $\Pi^0_3$ in any
universal left-r.e.\ numbering.
\end{proof}

\noindent
We remark that any set that can be made into itself via an acceptable
numbering contains
an infinite recursive subset by the Padding Lemma \citep{Odi89,Rog58}.
This means that the Martin-L\"of randoms, the recursively random sets,
the Schnorr randoms, the Kurtz randoms, the bi-immune sets, and immune
sets cannot be made into themselves using an acceptable numbering. 
Figueira, Miller and
Nies \cite{FMN09} asked whether Chaitin's $\Omega$ can have an
infinite co-r.e.\ indifferent set.  A partial solution to this problem
follows immediately from the Lemma~\ref{K-recursive into itself} and
Proposition~\ref{prop:left-r.e. ML not=self}: if such a co-r.e.\
indifferent set exists, it cannot be retraceable by a recursive function.

In contrast to Proposition~\ref{prop:left-r.e. ML not=self}, every
acceptable numbering of the left-r.e.\ reals makes the autoreducible
reals into themselves as the resulting index set is a
cylinder and thus autoreducible; the same applies for the notion
of strongly infinitely-often autoreducible sets as cylinders have
also that property. Note that not every set is
autoreducible, for example Martin-L\"of random reals fail to be
autoreducible \citep{FMN09,Tra70}. By Corollary~\ref{randoms are
random2}, the non-autoreducible reals can also be made into
themselves, but by the above comment they cannot be made into
themselves via an acceptable numbering.

\section{Singleton classes}

\noindent
For the case  of singletons, we can characterise which things can be
made into themselves.

\begin{thm} \label{thm: bambam}
A left-r.e.\ class $\{A\}$ can be made into itself iff $A \neq \emptyset$
and there exists an infinite, r.e.\ set $B$ such that $A \intersect B
= \emptyset$.
\end{thm}

\begin{proof}
Assume $A$ can be made into itself via a universal left-r.e.\
numbering $\alpha$.  Then $A \notin \{\emptyset, \omega\}$, so there
exists a rational number $r$ with $.A < r < 1$ where ``.A'' is the set
$A$ interpreted as a real number between 0 and 1.  Let
\[
B = \{e : (\exists s) [\alpha_{e,s} > r] \}.
\]
Then $A \intersect B = \emptyset$, $B$ is r.e., and $B$ is infinite.

Conversely, assume $A \neq \emptyset$, and $B$ is an infinite r.e.\
set satisfying $A \intersect B = \emptyset$.  Then $B$ has an infinite
recursive subset  $R = \{b_0, b_1, \dotsc\}$
Brodhead and Kjos-Hanssen \citep{BK09} showed that there exists a
\emph{Friedberg numbering}, or enumeration without repetition, of the
left-r.e.\ reals.  Let $\alpha$ be a Friedberg numbering of the
left-r.e.\ reals with the real $A$ deleted from the enumeration.

If $A$ is a finite set whose maximum element is $m$, then we can
hardwire $A$ into the numbering $\gamma$ as follows:
\[
\gamma_e = 
\begin{cases}
A & \text{if $e \in A$,} \\
\emptyset & \text{if $e \notin A$ and $e < m$,} \\
\alpha_{e - (m+1)} & \text{if $e > m$}.
\end{cases}
\]
Then $\gamma$ makes $A$ into itself.  Now assume $A$ is infinite, and
let $A_0, A_1, A_2, \ldots$ be a recursive approximation of $A$ from
below where $A_n \neq A$ for all $n$. We then build a further
numbering $\gamma$ such that
\[
  \gamma_e =
\begin{cases}
     \alpha_d & \text{if $e = b_d$,} \cr
     A_s      & \text{if $e \in \complement{A} \cap \complement{R}$
and $s = \max\{t: e \in A_t\}$,} \cr
     A        & \text{if $e \in A \cap \complement{R}$.} \cr
\end{cases}
\]
This $\gamma$ witnesses that $\{A\}$ can be made into itself. 
Moreover, $\gamma_e$ is left-r.e.\ via the following algorithm. 
Recursively decide whether the first case above is satisfied, and if
it is not then $\gamma_e$ follows the left-r.e.\ approximation for $A$
whenever it appears that $e \in A$.
\end{proof}

\noindent
In canonical universal left-r.e.\ numberings, no set gets made into itself.

\begin{prop}
Let $\alpha$ be an acceptable universal left-r.e.\ numbering.  Then
for every set $B$, $\{e : \alpha_e = B\} \neq B$.
\end{prop}

\begin{proof}
Every finite set has an infinite index set and is thus not made into itself.
For every infinite set consider the left-r.e.\ numbering $\beta$ given by
\[
\beta_e = B \cap \{x: (\exists y \in W_e)\: [ x < y ]\}.
\]
Note that $\beta_e = B$ iff $W_e$ is infinite and that there is a
recursive function $f$ with $\alpha_{f(e)} = \beta_e$ for all $e$.  It
follows that $W_e$ is infinite iff $\alpha_{f(e)} = B$.  Hence $\{e :
\alpha_e = B\}$ is not left-r.e.\ but rather $\Pi^0_2$-complete like
the index set for the infinite sets \cite{Soa87}.
\end{proof}

\section{Making things into themselves simultaneously}

\noindent
Having made certain classes into themselves and others not, we now
investigate which collections of classes can be simultaneously made
into themselves using a single numbering.

\begin{defn}
We say that $\A$ and $\B$ can be \emph{simultaneously} made into
themselves if there is a numbering which makes both $\A$ into itself
and $\B$ into itself.  
\end{defn}

\noindent
One thing we do not get at the same time is Martin-L\"of random
sets and weakly 1-generic sets.  We showed in Corollary~\ref{randoms
are random2} and Corollary~\ref{cor: bobo} that each of these classes
can be made into themselves (by themselves), however their combination
results in calamity.

\begin{prop} \label{prop: ML 1-generic calamity}
The Martin-L\"of random sets and weakly 1-generic sets cannot
simultaneously be made into themselves.
\end{prop}

\begin{proof}
Assume that $\alpha$ makes the weakly 1-generic sets into themselves. 
 Then the characteristic sequence for the $\alpha$-index set of the
weakly 1-generic sets is itself weakly 1-generic and hence must
contain very long runs of 1's \cite[Theorem~3.5.5]{Nie09}.  On the other
hand, no Martin-L\"of random sets is weakly 1-generic
\cite[Proposition~8.11.9]{DH10}, and therefore the $\alpha$-index set
for the Martin-L\"of random sets must contain very long runs of 0's. 
Thus it follows from \cite[Theorem~3.5.21]{Nie09}, which says that
long runs of 0's prevent a set from being Martin-L\"of random, that
the Martin-L\"of random sets do not get made into themselves using $\alpha$.
\end{proof}

\noindent
We note that for many classes which can be made
into themselves and which have complementary classes which can also be
made into themselves, the class and its complementary class cannot be
simultaneously made into themselves.

\begin{prop} \label{prop: simultaneous complement}
Any class closed under complements cannot be simultaneously made into
itself with its complement.
\end{prop}

\begin{proof}
Suppose that some class which is closed under complements can be made
into itself.  Then the indices for the complement in any universal
left-r.e.\ numbering are also a member of the original class and hence
do not belong to its complement.
\end{proof}

\noindent
Examples of important classes for which Proposition~\ref{prop:
simultaneous complement} applies
include the Martin-L\"of random sets and the autoreducible sets. 
Corollary~\ref{randoms are random2} established that the
Martin-L\"of random sets can be made into themselves, and any
acceptable universal left-r.e.\ numbering will make the
non-Martin-L\"of random sets into themselves via the Padding Lemma
\cite{Soa87}.  We established in the discussion following
Proposition~\ref{prop:left-r.e. ML not=self} that any acceptable
universal left-r.e.\ numbering also makes the autoreducible sets into
themselves.  Hence the following corollary holds.

\begin{cor}
The class of all sets which are not Martin-L\"of random and
the class of all autoreducible sets are
simultaneously made into themselves by any acceptable universal
left-r.e.\ numbering.
\end{cor}

\section{Minimal and maximal left-r.e.\ sets} \label{sec:minmax}

\noindent
A coinfinite r.e.\ set $A$ is called \emph{maximal} \cite{Fri58} iff there is
no coinfinite r.e.\ superset $E \supset A$ with $E-A$ being infinite;
in other words, an r.e.\ set $A$ is maximal iff $A \subset^* {\mathbb N}$
and there is no r.e.\ set $E$ with $A \subset^* E \subset^* {\mathbb N}$.
The corresponding notion of minimal r.e.\ sets does not exist. 
Indeed, every infinite r.e.\ set $A$ contains an infinite recursive
subset, and one can recursively remove every other element from this
infinite recursive set to obtain an infinite r.e.\ subset of $A$ with
infinitely fewer elements.

To what extent does the inclusion structure for the left-r.e.\ sets
resemble that of the r.e.\ sets?  One difference between these two
structures is immediate.  Unlike the situation for r.e.\ sets,
intersections and unions of left-r.e.\ sets need not be left-r.e.; only the join
\[
E \oplus F = \{2x: x \in E\} \cup \{2y+1: y \in F\}
\]
of left-r.e.\ sets $E$ and $F$ is always left-r.e.  For example,
$\Omega$ intersected with the set of even numbers, call this set $A$,
is not a left-r.e.\ set. If it were, then one could use this set to build a
left-r.e.\ approximation for the set $B = \{x : 2x \in \Omega\}$ by
updating at each stage those $B$-indices $e$ for which every odd
$A$-index below $2e$ shows a zero.  But, as established in
\eqref{eqn:2xOmega}, $B$ is low and Martin-L\"{o}f random,
contradicting that every left-r.e.\ Martin-L\"{o}f random is an
$\Omega$-number \cite{CHKW01,KS01,MT13} and that every $\Omega$-number is
weak-truth-table equivalent to, and
hence Turing equivalent to, the halting problem \cite{CN97}.
An analogous construction shows
that left-r.e.\ sets are not closed under inclusion.

\begin{defn} \label{def:zulu}
A left-r.e.\ set $A$ is called a \emph{minimal left-r.e.\ set}
iff $\emptyset \subset^* A$ and there is no left-r.e.\ set $E$
with $\emptyset \subset^* E \subset^* A$.
A left-r.e.\ set $B$ is called a \emph{maximal left-r.e.\ set}
iff $B \subset^* {\mathbb N}$ and there is no left-r.e.\ set $E$
with $B \subset^* E \subset^* {\mathbb N}$.
\end{defn}

\noindent
The next result shows that both types of sets exist, in contrast to
the r.e.\ case where only maximal sets exist. Neither maximal
left-r.e.\ sets, nor minimal left-r.e.\ sets, nor their respective
complements need be hyperimmune (in contrast to the complements of
maximal r.e.\ sets \cite[Proposition~III.4.14]{Odi89}).  

\begin{thm} \label{thm:zulu}
There are a minimal set $A$ and a maximal set $B$ in the partially
ordered structure of all left-r.e.\ sets and $\subset^*$.
\end{thm}

\begin{proof}
Let $\Omega$ be Chaitin's Martin-L\"of random set and let $\Omega_s$
be a left-r.e.\ approximation to it. Furthermore, let
\[
c_{n,s} = \sum_{m < 2^n} 2^{2^n-m} \Omega_s(m)
\]
and $c_n = \lim_{s\to\infty} c_{n,s}$.  Let $d_n = c_n - 2^{2^{n-1}} c_{n-1}$
so that
$d_n$ is the sum of all $2^{2^n-m} \Omega(m)$ with $m = 2^{n-1},2^{n-1+1},
\ldots,2^n-1$. Note that $c_n \leq 2^{2^n}$ for all $n$.
Let $I_1,I_2,\ldots$ be a
recursive partition of $\mathbb N$ into intervals such that each interval
$I_n$ contains all numbers $\langle n,x,y\rangle =
\min(I_n)+x \cdot 2^{2^n} + y$ with $x,y \in \{0,1,\ldots,2^{2^n}-1\}$.
Now let
\begin{eqnarray*}
  a_n & = & \langle n,c_{n-1},2^{2^n}-1-d_n\rangle \mbox{ for $n>0$}, \\
  b_n & = & g(a_n) \mbox{ where } \\
  g(u) & = &  \max(I_n)+\min(I_n)-u \mbox{ for all $n$ and all $u \in I_n$,} \\
  A & = &  \{a_1,a_2,\ldots\} \mbox{ and }
  B \ = \ {\mathbb N}-\{b_1,b_2,\ldots\}.
\end{eqnarray*}
So $g$ is defined such that if $u$ is the $r^\text{th}$ smallest
element of $I_n$
then $g(u)$ is the $r^\text{th}$ largest element of $I_n$. Note that $A$ and $B$
are btt-equivalent: $u \in A \Leftrightarrow g(u) \notin B$.
Now it is shown that $A$ is a minimal left-r.e.\ set and $B$ is a
maximal left-r.e.\ set.

The set $A$ is left-r.e.\ as one can start the enumeration at $s_0$
with $c_{0,s} = c_0$ and letting, for $s \geq s_0$,
$A_s = \{a_{1,s},a_{2,s},\ldots,a_{s,s}\}$. Then one
has for each $s \geq s_0$ that whenever there is 
an $n$ with $a_{n,s+1} > a_{n,s}$ then there is also
a least $m \leq n$ where $a_{m,s+1} \neq a_{m,s}$ and it follows
that for this number the change is in the $d$-part of
$a_{m,s} = \langle m,c_{m-1,s},2^{2^m}-1-d_{m,s}\rangle$ so that
$a_{m,s+1} < a_{m,s}$. Hence it holds that $A_s \leq_{\lex} A_{s+1}$
and the approximation of the $A_s$ is
an left-r.e.\ approximation. Furthermore, let $B_s = (I_1-\{b_{1,s}\})
\cup (I_2 - \{b_{2,s}\}) \cup \ldots \cup (I_s - \{b_{s,s}\})$.
Note that $g$ inverts the direction of the approximation in the intervals.
Hence, if $s \geq s_0$ and $b_{n,s+1} \neq b_{n,s}$ then the least $m \leq n$
with $b_{m,s+1} \neq b_{m,s}$ satisfies that $b_{m,s+1} > b_{m,s}$.
Hence one can see that for $s \geq s_0$ it holds that $B_s \leq_{\lex}
B_{s+1}$ and $\lim B_s = B$.

Assume now that $E$ is an infinite left-r.e.\ subset of $A$ and let $E_s$
be a left-r.e.\ approximation of $E$. For any $n$ where $a_{n+1} \notin E$
and $a_{n+2} \in E$, let $\sigma$ be an $n$-bit binary string telling which of
the first $n$ elements $a_1,\ldots,a_n$ is in $E$ and let
$\psi(\sigma,c_n)$ be a partial-recursive function identifying the
first stage $s \geq s_0$ such that $a_{1,s} = a_1$, $a_{2,s} = a_2$,
$\ldots$, $a_{n,s} = a_n$ and
$$
   E_s \cap J_{n+2} = 
   \{a_{m,s}: m \in \{1,2,\ldots,n\} \wedge
   \sigma(m)=1\} \cup \{a_{n+2,s}\};
$$
where $J_n = I_1 \cup I_2 \cup \ldots \cup I_n$.
Note that $n,a_1,\ldots,a_n$ can all be computed
from $c_n$. Now, due to $E_s \leq_{\lex} E$, the final
value of $a_{n+2}$ cannot exceed $a_{n+2,s}$ for the chosen~$s$,
hence $c_{n+1,s} = c_{n+1}$. This implies that for all the $n$ where
$a_{n+1} \notin E \wedge a_{n+2} \in E$ it holds that the Kolmogorov complexity
of $c_{n+1}$ given $c_n$ is at most $n$ bits plus a constant; however, the
prefix-free Kolmogorov complexity of each $c_n$ is approximately $2^n$
and therefore there can only be finitely many such $n$. It follows that
almost all $a_n$ are in $E$.
This shows that $A$ is a minimal left-r.e.\ set.

To see that $B$ is maximal, consider any coinfinite left-r.e.\ set $E$
containing $B$.
As before one computes for each $n$ with
$b_{n+1} \in E \wedge b_{n+2} \notin E$ and $\sigma$ being an
$n$-bit string telling which of $b_1,b_2,\ldots,b_n$ are in $E$
the stage $\psi(c_n,\sigma)$ as the first stage $s \geq s_0$ such that
$b_{1,s} = b_1$, $b_{2,s} = b_2$, $\ldots$, $b_{n,s} = b_n$ and
$$
   E_s \cap J_{n+2} = J_{n+2} - \{b_{m,s}: m \in \{1,2,\ldots,n\}
   \wedge \sigma(m) = 0\} - \{b_{n+2,s}\}.
$$
Note again that
$n,b_1,b_2,\ldots,b_n$ can be computed from $c_n$. Now
the $s = \psi(c_n,\sigma)$ satisfies that $b_{n+2,s} \leq b_{n+2}$
and hence $c_{n+1,s} = c_{n+1}$. This permits again to conclude
by the same Kolmogorov complexity arguments as in the case of the
set $A$ that $E$ is the union of $B$ and a finite set; hence
$B$ is a maximal left-r.e.\ set.
\end{proof}

\noindent
One might ask why we construct a maximal left-r.e.\ set instead of
checking whether some maximal r.e.\ set is also maximal as a left-r.e.\
set. Unfortunately this approach does not work, as the following result
shows.

\begin{thm} \label{thm: maxsep}
No r.e.\ set can be a maximal left-r.e.\ set.
\end{thm}

\begin{proof}
Let $A$ be an infinite r.e.\ set. Without loss of generality
assume that exactly one new element gets enumerated into $A$ at each
stage of its recursive approximation $A_0, A_1, A_2, \dotsc$ and for
each $s$, let
$x_0, x_1, x_2, \dotsc$ denote the complement of $A_s$ in ascending
order and define
\[
E_s = A_s \union \{x_1, x_3, x_5, \dotsc\}.
\]
Now assume that there is a stage $s$ and $x_n \in A_{s+1} - A_s$
being the unique element enumerated into $A$ at stage $s$.
If $n$ is even, then
\[
E_{s+1} = E_s \union \{x_n, x_{n+2}, x_{n+4}, \dotsc\} - \{x_{n+1},
x_{n+3}, \dotsc\},
\]
and if $n$ is odd, then
\[
E_{s+1} = E_s \union \{x_{n+1}, x_{n+3}, \dotsc\} - \{x_{n+2},
x_{n+4}, \dotsc\}.
\]
In either case the minimum of the symmetric difference of $E_s$ and
$E_{s+1}$, which is $x_n$ when $n$ is even and $x_{n+1}$ when $n$ is
odd, belongs to $E_{s+1}$.  Hence $E_s \leq_\lex E_{s+1}$.  The
left-r.e.\ set $E = \lim E_s$ contains all elements of $A$ and every
second element of the complement of $A$, hence $A$ is not maximal in
the structure of the left-r.e.\ sets under inclusion.
\end{proof}

\noindent
A further interesting question is the following: For maximal r.e.\ sets
$C$ one has the property that there is no r.e.\ set $E$ with $E-C$
and $\overline{E}-C$ being infinite \cite[p.~187]{Soa87}. Do the
corresponding properties
also hold for minimal and maximal left-r.e.\ sets? That is, can one
make sure that no left-r.e.\ set splits a minimal left-r.e.\ set $A$
into two infinite parts or the complement of a maximal left-r.e.\ set
$B$ into two infinite parts? The answer is ``no''.

\begin{thm}
Let $A$ be an infinite left-r.e.\ set and $B$ be a coinfinite left-r.e.\
set. Then there is an infinite left-r.e.\ set $E$ such that
$A \cap E$ and $A \cap \overline{E}$ are both infinite. Furthermore
there is an infinite left-r.e.\ set $F$ such that
$\overline{B} \cap F$ and $\overline{B} \cap \overline{F}$ are both
infinite.
\end{thm}

\begin{proof}
Assume by way of contradiction that $A$ and $B$ exist. Then the set
of even number neither splits $A$ nor the complement of $B$ into
two infinite halves; therefore without loss of generality, all
members of $A$ are odd and all non-members of $B$ are odd.

Let $A = \{a_0,a_1,a_2,\ldots\}$ and $\overline{B} = \{b_0,b_1,b_2,\ldots\}$
be denoted such that $a_k < a_{k+1}$ and $b_k < b_{k+1}$ for all $k$.
Now choose $E$ and $F$ such that
\begin{eqnarray*}
  E & = & \{a_{2k}, a_{2k+1}-1: k \in {\mathbb N} \} \mbox{ and } \\
  \overline{F} & = & \{b_{2k}, b_{2k+1}-1: k \in {\mathbb N} \}.
\end{eqnarray*}
One can obtain corresponding approximations $E_s$ and $F_s$ for $E$ and
$F$, respectively, by using analogous formulas to define $E_s$ from $A_s$
and $\overline{F_s}$ from $\overline{B_s}$. Fix left-r.e.\
approximations $A_s$ to $A$ with $A_s(2x)=0$ for all $x$ and $B_s$ to
$B$ with $B_s(2x) = 1$ for all $x$.  Then $A_s \leq_{\lex}
A_{s+1} \Rightarrow E_s \leq_{\lex} E_{s+1}$ and $B_s \leq_{\lex} B_{s+1}
\Rightarrow F_s \leq_{\lex} F_{s+1}$. Hence both sets $E$ and $F$
are left-r.e.\ sets.
Furthermore, $A \cap E = \{a_0,a_2,a_4,\ldots\}$,
$A \cap \overline{E} = \{a_1,a_3,a_5,\ldots\}$,
$\overline{B} \cap F = \{b_1,b_3,b_5,\ldots\}$
and $\overline{B} \cap \overline{F} = \{b_0,b_2,b_4,\ldots\}$.
Hence $E$ and $F$ meet the requirements.
\end{proof}

\noindent
Having established the fundamentals on minimal and maximal left-r.e.\
sets, the time is ready for the question which of them can be made
into themselves.

\begin{thm} \label{thm:lowerfarm}
There is a minimal left-r.e.\ set $A$ such that $\{A\}$ can be made
into itself. There is no maximal left-r.e.\ set $B$ such that
$\{B\}$ can be made into itself.
\end{thm}

\begin{proof}
One can easily see that the intervals $I_n$ in Theorem~\ref{thm:zulu} can be
chosen large enough so that $a_n \neq \max(I_n)$ for all $n$; hence
$A = \{a_0, a_1, \dotsc\}$ is disjoint from an infinite recursive set
and so $\{A\}$ can be made into itself by Theorem~\ref{thm: bambam}.

Assume now that $B$ is a maximal left-r.e.\ set; one has to show that there
is no infinite recursive set $R$ disjoint from $B$. Assume the contrary and
without loss of generality $R \cup B$ is coinfinite (otherwise $B$ is the
complement of a recursive set and not maximal).  Let $B_0,B_1,\ldots$
be a left-r.e.\ approximation of $B$.
Now one can select a sequence $s_0,s_1,\ldots$ of stages such
that $B_{s_t} \cap \{0,1,\ldots,t\}$ is disjoint from $R$.
Hence $E_t = (B_{s_t} \cap \{0,1,\ldots,t\}) \cup R$ is a
recursive left-r.e.\ approximation of $B \cup R$ which then
witnesses that $B$ was not, as assumed, a maximal left-r.e.\ set.
Hence there is no infinite recursive set disjoint to $B$ and,
by Theorem~\ref{thm: bambam}, $\{B\}$ cannot be made into itself.
\end{proof}

\noindent
The next result shows that each of the classes of minimal left-r.e.\ sets
and maximal left-r.e.\ sets cannot be made into itself; the proof method
is to show that the corresponding index-sets cannot be $K'$-recursive
and therefore cannot be left-r.e., let alone minimal or maximal.
\begin{thm}
Neither the class of minimal left-r.e.\ sets nor the class of maximal
left-r.e.\ sets can
be made into itself.
\end{thm}

\begin{proof}
Let $A$ be the minimal and $B$ be the maximal left-r.e.\ set from
Theorem~\ref{thm:zulu}. Recall that $I_1,I_2,\ldots$ is a recursive partition
of the natural numbers such that $A$ has exactly one element in $I_n$
for each $n$. Let $\ind(x)=n$ for the unique $n$ with $x \in I_n$; the
function $\ind$ is recursive.  We show that with respect to any
universal left-r.e.\ numbering $\alpha$, neither the minimal nor the
maximal left-r.e.\ sets can be made into itself.

Let $P$ be the index set of the minimal left-r.e.\ sets in $\alpha$. 
Now consider for any r.e.\ set $W_e$ the set $\tilde A_e$ given as
$$
    \{3x: x \in A \wedge \ind(x) \in W_e\} \cup
                \{3x+1,3x+2: x \in A \wedge \ind(x) \notin W_e\}.
$$
One can easily see that $\tilde A_e$ has a left-r.e.\ approximation;
starting with a left-r.e.\ approximation $A_s$ for $A$ and an enumeration
$W_{e,s}$ for $W_e$, the approximation $\tilde A_{e,s}$ is the same as
for $\tilde A_e$ except $A$ is replaced with $A_s$ and $W_e$ is
replaced with $W_{e,s}$.

If $W_e$ is cofinite then the set $\tilde A_e$ is a finite variant of
$\{3x: x \in A\}$ and thus minimal; if $W_e$ is coinfinite then the set
$\tilde A_e$ has an infinite left-r.e.\ subset which has infinitely many
less elements than $\tilde A_e$, namely
$$
   \{3x: x \in A \wedge \ind(x) \in W_e\} \cup
                \{3x+1: x \in A \wedge \ind(x) \notin W_e\}.
$$
There is a $K'$-recursive mapping which determines for every $e$ the
least index $d$ with $\alpha_d = \tilde A_e$; now $d \in P$ iff
$W_e$ is cofinite. As the set $\{e: W_e$ is cofinite$\}$ is not
$K'$-recursive in any acceptable numbering of the
r.e.\ sets \cite[Corollary~IV.3.5]{Soa87}, $P$ cannot be $K'$-recursive and
therefore is not a minimal left-r.e.\ set.

Now let $Q$ be the index set of the maximal left-r.e.\ sets in the given
enumeration $\alpha$. Recall that $B$ is a fixed maximal left-r.e.\ set.
Now each join $B \oplus W_e$ is left-r.e.\ and is a maximal
left-r.e.\ set iff $W_e$ is cofinite. Again there is a $K'$-recursive mapping
which finds for each $e$ an index $d$ with $B \oplus W_e = \alpha_d$;
hence one can, relative to $K'$, many-one reduce the index set of the
cofinite sets to $Q$. As the index set of the cofinite sets is not
$K'$-recursive, $Q$ also cannot be $K'$-recursive; hence $Q$ cannot
be left-r.e.\ and in particular is not a maximal left-r.e.\ set.
\end{proof}

\section{Inclusion}

\noindent
We now turn our attention to the question of which things can be directly stuck
inside other things.  Kummer \citep{Kum10} showed that there exists a
numbering $\phe$ of the partial recursive sets such that the r.e.\
\emph{inclusion problem},
\[
\INC_\phe = \{\pair{i,j}: W_i^\phe \subseteq W_j^\phe\},
\]
is recursive in the halting set and asked whether there exists a
numbering $\phe$ of the partial recursive sets such that $\INC_\phe$
is r.e. Kummer's question remains open, however in the context of
left-r.e.\ sets we show the answer is
negative.  Below we use $\INC_\alpha$ to denote the left-r.e.\
inclusion problem.

\begin{thm} \label{thm:left-r.e. inclusion}
For every universal left-r.e.\ numbering $\alpha$,
\begin{enumerate}[\scshape (i)]
\item $\INC_\alpha$ is not r.e.\ and
\item $\INC_\alpha \geq_\T K$.
\end{enumerate}
\end{thm}

\begin{proof}
For part~(\textsc{i}), define the following two sets:
\begin{align*}
 A &= \text{ the set of odd numbers}, \\
 B &= \{2x: x \in K\} \union \{2x+1: x \notin K\}.
\end{align*}
Note that $A \intersect B = \{2x+1: x \notin K\}$ and that $A$ and $B$
are both left-r.e.: the characteristic function of $B$ on $2x,2x+1$
changes from $01$ to $10$ whenever $x$ goes into $K$, hence this is a
left-r.e. process.

Let $\alpha$ be a universal left-r.e.\ numbering and suppose that
$\INC_\alpha$ were r.e.  For each number $x$, we show how to decide
membership in the set $\{y \in K: y < x\}$.  We search for a
left-r.e.\ set $E$ and a number $s$ such that the following has
happened up to stage $s$:
\begin{itemize}
\item The indices for $E \subseteq A$ and $E \subseteq B$ have both
been enumerated into the inclusion problem;

\item for all $y < x$, either $y \in K$ or $2y+1 \in E_s$ but not both.
\end{itemize}
Note that $E$ cannot acquire any further element $2z+1 < 2x$ after stage~$s$ as
then $2z+1\in B$ which implies $z \notin K$, contrary to the second
item above.  Hence $E$ does not change below $2x$ after stage $s$ and
therefore one knows for all $y < x$ that $y\in K$
iff $y \in E_s$.  An $\alpha$-index for such a set $E$ exists as every
finite set has an index in $\alpha$, and therefore our search
terminates.  The recursive algorithm just described thus decides the
halting problem, which is impossible.

For part~(\textsc{ii}), note that instead of searching for enumerations of the
inclusion problem, one can run the above algorithm relative to the inclusion
problem and so show that $K$ is Turing reducible to the inclusion problem
with that algorithm.
\end{proof}

\noindent
We leave the following open questions for the left-r.e.\ inclusion problem:

\begin{ques}
Does there exist a numbering $\alpha$ for the left-r.e.\ sets such
that $\INC_\alpha \equiv_\T K$?  In particular, can we make
$\INC_\alpha$ to be left-r.e.?
\end{ques}

\noindent
Consider the related relation
\[
\LEX_\alpha = \{\pair{i,j} : \alpha_i \leq_\lex \alpha_j\}.
\]
Any Friedberg numbering $\alpha$ makes $\LEX_\alpha$ recursive in the
halting set.  The reason is that no two distinct indices in a
Friedberg numbering represent the same left-r.e.\ set, so a halting
set oracle suffices to find a sufficiently long prefix which reveals
the lexicographical order of the strings.  We can improve this result to
a numbering such that the left-r.e.\ relation itself becomes left-r.e.

\begin{thm} \label{gazebo}
There exists a universal left-r.e.\ numbering $\alpha$ such that
$\LEX_\alpha$ is an r.e.\ relation.
\end{thm}

\begin{proof}
Let $\beta$ be a Friedberg left-r.e.\ numbering which includes indices
for all the left-r.e.\ sets except for $\mathbb N$.  We define a
universal left-r.e.\ numbering $\alpha$ based on $\beta$ as follows. 
Informally, during the first $s$~stages, $\alpha$ follows the first
$s$~indices of $\beta$ for $s$~computation steps, and some finitely
many other $\alpha$-indices $e$ have been defined to be $\alpha_e  =
\mathbb N$.  If $\alpha_e = {\mathbb N}$, we say that the index $e$
has been \emph{obliterated}.  We describe stage~$s+1$.  For each pair
$\pair{i,j}$ with $i<j$ where $\beta_i$ becomes lexicographically
larger than $\beta_j$ at stage~$s+1$, that is, $\beta_{i,s} \leq_\lex
\beta_{j,s}$ but $\beta_{i,s+1} >_\lex \beta_{j,s+1}$, the index for
the $\alpha$-follower of $\beta_j$ and all larger defined
$\alpha$-indices are obliterated and a new $\alpha$-follower for
$\beta_j$ and each of the other newly obliterated indices is
established.  Also in stage~$s+1$, an $\alpha$-follower for
$\beta_{s+1}$ is established so that in the end each $\beta$-index
will have a unique $\alpha$-index following it.  Note that only
finitely many $\alpha$-indices are defined in any given stage.

For every $e$, the $\alpha$-index following $\beta_e$ eventually
converges once sufficiently much time has passed to allow the
approximation of $\beta_e$'s prefix to differ from the approximation
of every lesser $\beta$-index's prefix and also enough time that these
prefixes never again change.  Furthermore, obliterating indices can
only ever increase membership of the respective set, so $\alpha$ is a
universal left-r.e.\ numbering.  Finally, $\alpha$ is r.e.\ because
whenever $\beta$'s enumeration tries to push $\pair{i,j}$ out of
$\LEX_\alpha$, the index $j$ gets obliterated and hence $\pair{i,j}$
stays inside $\LEX_\alpha$.
\end{proof}

\section*{Summaries}
\hyphenation{Selbst-be-zugs-eigen-schaf-ten}

\noindent
\begin{CJK}{GB}{gbsn}
\bf{关于可以成为自己一员的性质.}
\end{CJK}\begin{CJK}{GB}{gbsn}
我们说一个自然数的性质P可以成为自己的一员，指的是所有左递归集有一个编码使得满足性质P的指标集也具有性质P。 例如，Martin-Lof随机性质就可以变成自己的一员。 在此，我们刻画所有可以成为自己一员的单元集性质。我们接着研究，有限同余情况下，左递归集所组成的类在包含关系下的结构。这种结构不仅有极大元而且有极小元。相比而言，相应的递归集所组成的类只有极大元没有极小元。 而且，我们构造左递归集的极大元和极小元的方法与经典的Friedberg关于递归类的极大元的方法有很大不同。最后，本文研究极大和极小左递归集的性质是否可以变成自己的一员。
\end{CJK} 

\medskip

\noindent
{\bf A\^\j oj kiojn oni povas meti en si mem.}
Aro $A$ estas rekursive enumerabla se $A$ estas la limo de uniforme rekursivaj
aroj $A_0,A_1,\ldots$ je kiuj $A_n \subseteq A_{n+1}$ por \^ciu $n$;
$A$ estas maldekstre rekursive enumerabla se $A$ estas la limo de uniforme
rekursivaj aroj $A_0,A_1,\ldots$ je kiuj $A_n \leq_{lex} A_{n+1}$ por
\^ciu $n$.
La publika\^\j o temas pri la sekvanta afero: Se $\alpha_0,\alpha_1,\ldots$
estas numerado da maldekstre rekursive enunmerablaj aroj kaj se
$P$ estas abstrakta eco de aroj (kiel esti Martin-L\"of hazarda),
tiam oni konsideru la indeksa aro $\{e: \alpha_e$ havas econ $P\}$. Oni diras
ke oni povas meti la $P$ en si mem se ekzistas numerado
$\alpha_0,\alpha_1,\ldots$ de {\^c}iuj maldekstre rekursive enumerablaj
aroj tiel ke la indeksa aro por $P$ je tiu numerado anka{\v u}
havas la econ $P$. En tiu-\^ci publika\^\j o estas diversaj teoremoj
kiuj diras je multaj famaj ecoj el teorioj pri rekursivaj funkcioj kaj
algoritmika hazardo se oni povas meti tiujn ecojn en si mem.
Ekzemple, oni povas meti la Martin-L\"of hazarda arojn en si mem.
Plue, se la aro $A$ havas minimume unu membron kaj estas maldekstre
rekursive enumerabla, tiam oni povas meti
la econ $P(X)$ dirante $X = A$ en si mem ekzakte se ekzistas malfinia
rekursive enumerebla aro $B$ kiu havas malplenan komuna\^\j on kun $A$.
Oni anka{\v u} esploras pri minimumaj kaj maksimumaj aroj en la
strukturo de maldekstre rekursive enumerablaj aroj je la ordo $\subseteq^*$.
Kvankam en la mondo de rekursive enumerablaj aroj la minimuma aroj ne ekzistas,
amba{\v u} ekzistas en la mondo de maldekstre rekursive enumerablaj aroj
kaj la pruvo malsimilas al tiu de Friedberg por la mondo de rekursive
enumerablaj aroj.

\medskip
\noindent
{\bf Dinge die in sich selbst gemacht werden k\"onnen.}
Eine Menge $A$ nat\"urlicher Zahlen heisst rekursiv aufz\"ahlbar (r.a.)
genau dann wenn es eine uniform-rekursive Folge $A_0,A_1,\ldots$ gibt
welche punktweise gegen $A$ konvergiert und $A_n \subseteq A_{n+1}$
f\"ur alle $n$ erf\"ullt; $A$ heisst links-r.a.\ genau dann wenn es eine
uniform-rekursive Folge $A_0,A_1,\ldots$ gibt welche punktweise gegen $A$
konvergiert und $A_n \leq_{lex} A_{n+1}$ f\"ur alle $n$ erf\"ullt.
Das Thema der Arbeit ist der folgende Selbstbezug: Man sagt dass eine
Eigenschaft $P$ von Mengen nat\"urlicher Zahlen
in sich selbst gemacht werden kann wenn es eine
Numerierung $\alpha_0,\alpha_1,\ldots$ aller links-r.a.\ Mengen gibt so dass
die Index-Menge $\{e: \alpha_e$ hat die Eingenschaft $P\}$ ebenfalls die
Eigenschaft $P$ hat. Es wird untersucht, welche bekannten
rekursions-theoretischen Eigenschaften diese Art von Selbstbezug haben,
zum Beispiel hat die Eigenschaft ``Martin-L\"of zuf\"allig'' einen solchen
Selbstbezug. Man kann auch die Eigenschaft $P$ betrachten wo $P(X)$
bedeutet dass $X = A$ ist f\"ur eine feste gegebene nichtleere
links-r.a.\ Menge $A$.
Nun hat $P$ die obenerw\"ahnte Art von Selbstbezug genau dann
wenn $A$ zu einer unendlichen rekursiv aufz\"ahlbaren
Menge $B$ disjunkt ist. Desweiteren wurde die Struktur der
links-r.a.\ Mengen mit der partiellen Ordnung $\subseteq^*$ untersucht.
Es wird gezeigt dass es in dieser Struktur, anders als im Fall der r.a.\ Mengen,
nicht nur maximale sondern auch minimale links-r.a.\ Mengen gibt; die
Konstruktion ist recht unterschiedlich von der Konstruktion welche Friedberg
im r.a.\ Fall benutzte. Desweiteren werden die Selbstbezugseigenschaften
von minimalen und maximalen links-r.a.\ Mengen untersucht.

\section*{Acknowledgments}
\noindent
The authors would like to thank Randall Dougherty for pointing out
that the class of finite sets cannot be made into itself. We thank Chunlai Zhou for writing the Chinese summary for us.

\end{document}